\documentclass[11pt]{amsart}

\usepackage[english]{babel}
\usepackage{amsmath,amssymb,amsthm}
\usepackage[pdftex]{graphicx,color}

\newtheorem{theorem}{Theorem}[section]

\newtheorem{lemma}[theorem]{Lemma}
\newtheorem{corollary}[theorem]{Corollary}

\theoremstyle{definition}
\newtheorem{defn}[theorem]{Definition}
\newtheorem{example}[theorem]{Example}

\newcommand{\G}{\mathcal{G}}
\newcommand{\F}{\mathcal{F}}
\newcommand{\A}{\mathcal{A}}
\newcommand{\B}{\mathcal{B}}
\renewcommand{\H}{\mathcal{H}}
\newcommand{\T}{\mathcal T}                    
\newcommand{\E}{\mathbb{E}}
\newcommand{\C}{\mathbb{C}}
\newcommand{\QCE}[3]{\E_{#1}\left[ {#2} \,|\, {#3} \right]} 
\newcommand{\QE}[2]{\E_{#1}\left[ {#2}\right]} 
\newcommand{\ch}[1]{\chi_{#1}}
\newcommand{\tr}{ \operatorname{Tr} } 
\newcommand{\ran}{ \operatorname{Ran}} 
\newcommand{\state}[1]{\mathcal{S}(#1)}
\newcommand{\eff}[1]{\operatorname{Eff}(#1)}
\newcommand{\ac}{ \ll_{\rm ac}}
\newcommand{\dd}{\mathrm{d}}
\renewcommand{\d}{\, \mathrm{d}}
\newcommand{\borel}[1]{\mathcal{O}(#1)}

\newcommand{\geo}{\#}
\newcommand{\eps}{\varepsilon}
\newcommand{\ds}{\displaystyle}

\newcommand{\BH}{\B(\H)}
\newcommand{\SH}{S(\H)}

\begin{document}

\title{On a quantum martingale convergence theorem}

\author{Kyler S.~Johnson}
\address{\color{black}Department of Mathematics and Statistics, University of Regina\\ Regina, SK S4S 0A2 Canada\color{black}\\
\color{black}johnskyl@uregina.ca\color{black}}

\author{Michael J.~Kozdron}
\address{Department of Mathematics and Statistics, University of Regina\\ Regina, SK S4S 0A2 Canada\\
michael.kozdron@uregina.ca}

\begin{abstract} 
It is well-known in quantum information theory that a positive operator valued measure (POVM) is the most general kind of quantum measurement. 
Mathematically, a quantum probability is a normalised POVM, namely
a function on certain subsets of a (locally compact and Hausdorff) sample space that satisfies the formal requirements for a probability measure and whose values are positive operators acting on a complex Hilbert space. A quantum random variable is an operator valued function which is measurable with respect to a quantum probability.
In the present work, we study quantum random variables and generalize several classical limit results to the quantum setting.  We prove a quantum analogue of the Lebesgue dominated convergence theorem and use it to prove a quantum martingale convergence theorem.  This quantum martingale convergence theorem is of particular interest since it exhibits non-classical behaviour; even though the limit of the martingale exists and is unique, it is not explicitly identifiable. However, we provide a partial classification of the limit through a study of the space of all quantum random variables having quantum expectation zero.  
\end{abstract}



\maketitle 


\section{Introduction}

Consider a quantum system possessing various physical properties. Using an experimental apparatus, some of these properties, known as observables, can actually be measured. To formulate the mathematics of quantum measurement, we model the states of the quantum system by density operators $\rho$ acting on a $d$-dimensional Hilbert space $\H$ and an observable by a hermitian operator. The experimental apparatus is represented by a positive operator valued measure (POVM). We can formally consider the experimental apparatus as a quantum probability $\nu$ (i.e., a positive operator valued \emph{probability} measure) acting on $(X, \F(X))$ satisfying $\nu(X)= 1$, the identity operator on $\H$, where $X$ denotes the sample space of possible outcomes of the measurement and $\F(X)$ is a suitable $\sigma$-algebra of events.

In practice, we usually take $X$ to be a finite set and $\F(X)$ to be the power set of $X$. Herein, we work in more generality by taking $X$ to be a locally compact Hausdorff space and $\F(X)$ to be a $\sigma$-algebra containing the Borel sets of X. The statistical/random facet of quantum measurement is captured by the following axiom: if, at the moment of the measurement, the system is in state $\rho$,  the probability that the event $E \in \F(X)$ will be measured is $\tr(\rho \nu(E))$.

While studying classical and non-classical convexity properties of the space of positive operator valued measures on  $(X,\F(X))$  with values in $\B(\H)$, 
the algebra of linear operators acting on $\H$, a transform was introduced in Ref.~\cite{farenick--plosker--smith2011} that associates any positive operator valued measure $\nu$ with a certain completely positive linear map $\Gamma(\nu)$ of the homogeneous C*-algebra $C(X) \otimes \B(\H)$ into $\B(\H)$. 

This association was achieved by using an operator valued integral in which operator valued functions are integrated with respect to positive operator valued measures and which has the feature that the integral of a random quantum effect is itself a quantum effect.

In Ref.~\cite{farenick--kozdron2012}, a better mathematical understanding of quantum probability was proposed through the introduction of a quantum analogue for the expected value $\QE{\nu}{\psi}$ of a quantum random variable $\psi$ relative to a quantum probability measure $\nu$ using the operator valued integral of Ref.~\cite{farenick--plosker--smith2011}. This led to theorems for a change of quantum measure and a change of quantum variables. Also introduced was a quantum conditional expectation which resulted in quantum versions of some standard identities for Radon-Nikod\'ym derivatives, and led to the formulation and proof of a quantum analogue of Bayes' rule.

It is a basic fact of functional analysis that if $\psi:X\rightarrow\C$ is an essentially bounded function on a probability space $(X,\F(X), \mu)$, then the
essential range of $\psi$ is precisely the spectrum of $\psi$, where one considers $\psi$ as an element of the
von Neumann algebra $L^\infty(X,\mu)$. Recently, a similar result for essentially bounded quantum random variables on quantum probability spaces using higher dimensional spectra was found; see~\cite{FKP} This investigation of quantum variance also involved notions from spectral theory, and it was discovered that the quantum moment problem admits a characterisation entirely within spectral terms.

In the present work, we build on these earlier results by considering for the first time limiting operations for sequences of quantum random variables and quantum probability measures including a quantum analogue of the Lebesgue dominated convergence theorem and a discrete Fubini-type theorem.  As in those earlier investigations, the noncommutativity of operator algebra leads to some structure that simply does not appear in the classical setting. Using the quantum conditional expectation of Ref.~\cite{farenick--kozdron2012}, we also establish a quantum martingale convergence theorem for quantum martingales obtained by conditioning on a fixed quantum random variable. This theorem is of particular interest since it strongly exhibits non-classical behaviour; even though the limit of the martingale exists and is unique, it is not identifiable. However, we provide a partial classification of the limit through a study of the space of quantum random variables having quantum expectation zero.   The outline of the paper is as follows. In Section~\ref{Introsect} we introduce our notation and summarize the relevant results of Refs.~\cite{farenick--kozdron2012, farenick--plosker--smith2011}, and~\cite{FKP}. We provide our first limiting results in Section~\ref{QEsect} and then study quantum random variables having quantum expectation zero in Section~\ref{MeanZerosect}.  Finally, in Section~\ref{MCTsect} we develop our quantum martingale convergence theorem.

It is worth mentioning two related papers. 
In the present paper, we generalize the Lebesgue dominated convergence theorem from the classical setting to the quantum setting and apply it to prove a quantum analogue of the martingale convergence theorem. Using somewhat related techniques, Ref.~\cite{PR19}  generalizes Lyapunov's convexity theorem for classical (scalar-valued) measures to quantum (operator-valued) measures.  And in Ref.~\cite{MPR20}, the object of study is 
positive operator valued measures whose image is the bounded operators acting on an infinite-dimensional Hilbert space and, when possible, the usual assumption of positivity of the operator valued measure is relaxed. The literature for such POVMs on infinite-dimensional Hilbert spaces is not as mature as the case of a finite dimensional Hilbert spaces and, consequently, it is currently unknown whether the present results could be extended to the infinite-dimensional setting.

\section{Notation and background results}\label{Introsect}

We will always write $\H$ for a $d$-dimensional Hilbert space, $\B(\H)$ for the
C$^*$-algebra of linear operators acting on $\H$, and  $\B(\H)_+$ for the cone of positive operators.
The predual of $\B(\H)$ is denoted by $\T(\H)$, the space of trace-class operators. 
Since $\H$ is finite dimensional, $\B(\H)$ and $\T(\H)$
coincide as sets. Finally, $X$ shall denote a locally compact Hausdorff space
and $\F(X)$ a $\sigma$-algebra of subsets of $X$ containing the Borel sets. In particular, $\borel{X}$, the  Borel sets of $X$, is itself a $\sigma$-algebra of interest.  
A density operator, or state, on $\H$ is a positive trace-class operator $\rho$ such that $\tr(\rho)=1$; the set
of all density operators is denoted by $\state{\H}$.
By a quantum effect we mean a positive operator $h \in \B(\H)_+$ with the property that every eigenvalue $\lambda$ of $h$ satisfies $0 \le \lambda \le 1$, and we let $\eff{\H}$ denote the set of quantum effects.  Note that every
state $\rho \in \state{\H}$ is also a quantum effect.
A set function $\nu:\F(X)\rightarrow\B(\H)$ is called a positive operator valued measure (POVM) if
\begin{itemize}
\item[(i)] \color{black}
$\nu(E) \in\B(\H)_+$ for every $E \in \F(X)$,
\color{black}
\item[(ii)] $\nu(X) \neq 0$, and
\item[(iii)] for every countable collection $\{E_k\}_{k=1}^\infty \subseteq \F(X)$ with $E_j \cap E_k = \emptyset$ for $j \neq k$ we have
\begin{equation}\label{sum}
\nu\left(\bigcup_{k=1}^\infty E_k \right) = \sum_{k =1}^\infty \nu(E_k).
\end{equation}
\end{itemize}
If, in addition, $\nu(X)=1\in\B(\H)$, then $\nu$ is called a quantum probability measure. 
The convergence in~\eqref{sum} above is normally assumed to be with respect to the ultraweak topology; however, 
because $\B(\H)$ has finite dimension, the convergence in~\eqref{sum} may be taken with respect to any of the usual operator topologies on
$\B(\H)$. The POVM 
$\nu:\F(X)\rightarrow\B(\H)$ induces the classical (i.e., scalar valued) measure
$\mu$ via $\mu=(1/d)\tr\circ \nu$, where $\tr$ is the canonical trace on $\B(\H)$. Note that if $\nu$ is a quantum probability measure, then $\mu$ is a classical probability measure.
We call the triple $(X, \F(X), \nu)$ a quantum probability space.

A function $\psi:X\rightarrow \B(\H)$ is said to be measurable (i.e., a quantum random variable)
if, for every pair $\xi$, $\eta\in\H$, the complex valued function $x\mapsto\langle\psi(x)\xi,\eta\rangle$
is measurable (i.e., a  random variable) in the classical sense.  In fact, it is known~\cite{FKP} that 
 $\psi:X\rightarrow\B(\H)$  is  measurable if and only if $\psi^{-1}(U)$ is a measurable set, for every open set $U\subseteq\B(\H)$.

The predual of the von Neumann algebra 
$L^\infty(X,\mu)\,\overline\otimes\,\B(\H)$ is given by $L^1_{\T(\H)}(X,\mu)$;  see Theorem~IV.7.17 of Ref.~\cite{Takesaki-bookI}.
In particular, if
$\Psi\in L^\infty(X,\mu)\,\overline\otimes\,\B(\H)$, then there is a bounded quantum random 
variable $\psi:X\rightarrow\B(\H)$ such that, for each $f \in L^1_{\T(\H)}(X,\mu)$, the complex number
$\Psi(f)$ is given by
\[
\Psi(f) = \frac{1}{d} \int_X  \tr\left(f(x)\psi(x)\right) \d \mu(x).
\]
Although $\psi$ is not unique, it is unique up to a set of $\mu$-measure zero.
We therefore identify $\Psi$ and $\psi$ and consider the elements 
of $L^\infty(X,\mu)\,\overline\otimes\,\B(\H)$ to be bounded quantum random variables $\psi:X\rightarrow\B(\H)$.
Note that $ L^\infty(X,\mu)\,\overline\otimes\,\B(\H) \cong L^\infty(X,\mu)\otimes M_d(\C)$
where $M_d(\C)$ is the space of $d\times d$ matrices over $\C$.

We end this section by stating a number of theorems and definitions from Refs.~\cite{farenick--kozdron2012, farenick--plosker--smith2011}, and~\cite{FKP} relevant for our purposes.
Recall that if $\nu_1$ and $\nu_2$ are both positive operator valued measures on $(X,\F(X))$, then $\nu_2$ is absolutely continuous with respect to $\nu_1$, written
$\nu_2 \ac \nu_1$, if $\nu_2(E) = 0$ for every $E \in \F(X)$ with $\nu_1(E) = 0$.  Furthermore, if $\mu$ is a classical  measure, then we can always view $\mu$ as the scalar valued POVM $\mu \cdot 1$.

\begin{theorem}
If $\nu$ is a POVM on $(X,\F(X))$, then $\nu$ is absolutely continuous with respect to the induced classical 
measure $\mu$, and there exists an $\F(X)$-measurable function
 $\displaystyle \frac{\dd\nu}{\dd\mu}$ such that 
\begin{equation}\label{rn defn}
\displaystyle\int_E\left\langle\displaystyle\frac{\dd\nu}{\dd\mu}(x)\xi,\xi \right\rangle\d\mu(x)  = 
\langle \nu(E)\xi,\xi \rangle ,
\end{equation}
for all $E\in \F(X)$ and all $\xi\in \H$.
The function $\displaystyle\frac{\dd\nu}{\dd\mu}$  is
called the \emph{principal Radon-Nikod\'ym derivative of $\nu$} and is a positive operator for
$\mu$-almost all $x\in X$.    
\end{theorem}

\begin{defn}\label{nuintdefn}
A measurable function $\psi:X \to \B(\H)$ is $\nu$-integrable if for every density operator $\rho$ the complex valued function
\[
\psi_\rho(x) = \tr\left(\rho \left(\displaystyle\frac{\dd\nu}{\dd\mu}(x)\right)^{1/2}\psi(x)\left(\displaystyle\frac{\dd\nu}{\dd\mu}(x)\right)^{1/2}\right), \;x\in X,
\]
is $\mu$-integrable.
The integral of a $\nu$-integrable function $\psi:X\rightarrow\B(\H)$
is defined to be the unique operator acting on $\H$ having the property that
\[
\tr\left(\rho\int_X\psi\d\nu\right) = \int_X \psi_\rho\d\mu ,
\]
for every density operator $\rho$.
\end{defn}

\begin{theorem}
If $\nu_1$, $\nu_2$ are POVMs on $(X,\F(X))$, then  $\nu_2 \ac \nu_1$ if and only if there exists a bounded $\nu_1$-integrable $\F(X)$-measurable function  $\displaystyle \frac{\dd\nu_2}{\dd\nu_1}$, unique up to sets of $\nu_1$-measure zero, such that
$$\nu_2(E) = \int_E\frac{\dd\nu_2}{\dd\nu_1} \d \nu_1$$
for every $E \in \F(X)$. 
Moreover,
\[
\frac{\dd\nu_2}{\dd\nu_1}= \left(\frac{\dd\mu_2}{\dd\mu_1}\right)
\left[ 
\left(\frac{\dd\nu_1}{\dd\mu_1}\right)^{-1/2}\left(\frac{\dd\nu_2}{\dd\mu_2}\right)\left(\frac{\dd\nu_1}{\dd\mu_1}\right)^{-1/2}
\right]
\]
and is called the \emph{non-principal Radon-Nikod\'ym derivative of $\nu_2$ with respect to $\nu_1$}.
\end{theorem}

Recall from Refs.~\cite{KuboAndo} and~\cite{Pusz} that
if $a,b \in \B(\H)_+$ are both invertible, then the geometric mean of $a$ and $b$ is the positive operator $a \geo b$ defined by  
$a \geo b = a^{1/2} (a^{-1/2} ba^{-1/2})^{1/2} a^{1/2}$.
If $a$, $b \in \B(\H)_+$ are non-invertible, then $a \geo b$ is defined by
\[
a \geo b= \lim_{\eps\to0+} (a +\eps 1) \geo (b + \eps 1),
\]
with convergence in the strong operator topology. 
If $\nu_1$ and $\nu_2$ are both quantum probability measures with  $\nu_2 \ac \nu_1$ and  if $\psi: X \to \B(\H)$ is a 
quantum random variable, then we define
\begin{equation}\label{boxtimesdefn}
\psi \boxtimes \frac{\dd \nu_2}{\dd\nu_1}=
 \left(\left(\frac{\dd\nu_1}{\dd\mu_1}\right)^{-1}\geo\frac{\dd \nu_2}{\dd \nu_1} \right)\left(\frac{\dd\nu_1}{\dd\mu_1}\right)^{1/2} \psi \left(\frac{\dd\nu_1}{\dd\mu_1}\right)^{1/2} \left(\left(\frac{\dd\nu_1}{\dd\mu_1}\right)^{-1}\geo\frac{\dd \nu_2}{\dd \nu_1} \right).
 \end{equation}
In particular, 
\begin{equation}\label{boxtimesdefnparticular}
\psi\boxtimes\ds\frac{\dd\nu}{\dd\mu} =\left(\frac{\dd\nu}{\dd\mu}\right)^{1/2}\psi\left(\frac{\dd\nu}{\dd\mu}\right)^{1/2}.
\end{equation}

\begin{defn} 
If  $\nu:\F(X)\rightarrow\B(\H)$ is a quantum probability measure, then the quantum expectation of $\psi$ with respect to $\nu$ is the map 
$\E_{\nu}: L^\infty(X,\mu)\,\overline\otimes\,\B(\H) \rightarrow\B(\H)$ defined by
\[
\QE{\nu}{\psi} = \int_X\psi\d\nu.
\]
\end{defn}

Recall from Chapter~3 of Ref.~\cite{Paulsen-book} that a linear map $\varphi:\A\rightarrow\B$ of unital C$^*$-algebras is a unital completely positive (ucp) 
map if $\varphi(1_\A)=1_\B$ and
the induced linear maps
$\varphi\otimes{\rm id_n}:\A\otimes M_n(\C)\rightarrow\B\otimes M_n(\C)$
are positive for every $n\in\{1,2,\ldots\}$.

\newpage

\color{black}
The following theorem gives one important property of quantum expectation. It is Theorem~2.5 of  Ref.~\cite{farenick--plosker--smith2011} and one of the main results of that paper; see also  Theorem~2.3 of Ref.~\cite{farenick--kozdron2012}.
\color{black}

\begin{theorem}\label{varineq} Quantum expectation is a completely positive operation. That is, the linear map 
$\E_{\nu}: L^\infty(X,\mu)\,\overline\otimes\,\B(\H) \rightarrow\B(\H)$
is a ucp map, for every 
quantum probability measure $\nu$.
\end{theorem}

The following example carefully explains how  one can view $\QE{\nu}{\psi}$ as a quantum averaging of $\psi$. 
A version of this first appeared in Ref.~\cite{farenick--plosker--smith2011};  see also Theorem~2.3(4) of Ref.~\cite{farenick--kozdron2012}.

\begin{example}\label{quantumaverageexample} Let $X=\{x_1, \dots, x_n\}$ and let $\F(X)$ be the power set of $X$. 
If  $h_1, \dots, h_n\in \B(\H)_+$ are such that $h_1+\cdots + h_n=1\in \B(\H)$, and $\nu$ satisfies
$\nu(\{x_j\})=h_j$ for $j=1, \dots, n$, then for every $\psi:X\rightarrow \B(\H)$ we have
\[
\QE{\nu}{\psi}=\int_X\psi \d\nu=\sum_{j=1}^nh_j^{1/2}\psi(x_j)h_j^{1/2}.
\]
\end{example}


\section{Continuity of quantum expectation}\label{QEsect}

In this section we establish a natural quantum analogue of the classical Lebesgue dominated convergence theorem, namely Theorem~\ref{DCT}, continuity of quantum expectation, along with some related results.

\begin{defn}
Let $\psi:X \to \BH$ and suppose that  $\{\psi_n\}_{n=1}^\infty$ is a sequence of quantum random variables. We say
$\psi_n$ converges ultraweakly $\mu$-almost surely to $\psi$ if $\tr(\rho\psi_n(x))\to\tr(\rho\psi(x))$ for all $\rho\in\SH$ and $\mu$-almost all $x\in X$.  
\end{defn}

It is an easy fact that the ultraweak $\mu$-almost sure limit $\psi$ of the previous definition is itself a quantum random variable.

\begin{lemma}\label{QConv}
Let $\psi:X \to \BH$ and suppose that  $\{\psi_n\}_{n=1}^\infty$ is a sequence of quantum random variables. If $\psi_n$ converges ultraweakly $\mu$-almost surely to $\psi$, then $\psi$ is a quantum random variable.  
\end{lemma}

\begin{proof}
Since $\psi_n$ converges ultraweakly $\mu$-almost surely to $\psi$, it follows that $\tr(\rho\psi_n(x))\to\tr(\rho\psi(x))$ for all $\rho\in S(\H)$ and $\mu$-almost all $x\in X$.  But since each $\tr(\rho\psi_n(x))$ is a complex valued random variable, the limit of the sequence $\{\tr(\rho\psi_n(x))\}_{n=1}^\infty$ converges to a complex valued random variable, namely $\tr(\rho\psi(x))$ for each $x\in X$, and therefore $\psi$ is a quantum random variable.  
\end{proof}
		
\begin{lemma}\label{QConvCor}
Let $\{\psi_n\}_{n=1}^\infty$ be a sequence of quantum random variables. If $\psi_n$ converges ultraweakly $\mu$-almost surely to $\psi$, then $\psi_n\boxtimes\dfrac{\dd\nu}{\dd\mu}$ converges ultraweakly $\mu$-almost surely to $\psi\boxtimes\dfrac{\dd\nu}{\dd\mu}$.  
\end{lemma}

\begin{proof}
For $\rho \in \SH$ and $x \in X$, let 
$$\tilde{\rho}_x=\left[\tr\left(\rho\dfrac{\dd\nu}{\dd\mu}(x)\right)\right]^{-1}\left(\left(\dfrac{\dd\nu}{\dd\mu}(x)\right)^{1/2}\rho\left(\dfrac{\dd\nu}{\dd\mu}(x)\right)^{1/2}\right),$$
and notice that $\tilde{\rho}_x\in\SH$.  
\color{black}
We know from~\eqref{boxtimesdefnparticular} that
\begin{align*}
\rho &\left(\psi_n\boxtimes\ds\frac{\dd\nu}{\dd\mu}\right)(x)\\
&=\rho\left(\frac{\dd\nu}{\dd\mu}(x)\right)^{1/2}\psi_n(x)\left(\frac{\dd\nu}{\dd\mu}(x)\right)^{1/2}\\
&=\tr\left(\rho\dfrac{\dd\nu}{\dd\mu}(x)\right)
\left[\tr\left(\rho\dfrac{\dd\nu}{\dd\mu}(x)\right)\right]^{-1}
\rho\left(\frac{\dd\nu}{\dd\mu}(x)\right)^{1/2}\psi_n(x)\left(\frac{\dd\nu}{\dd\mu}(x)\right)^{1/2}
\end{align*}
and so using properties of the trace functional, we therefore obtain
\begin{align*}
\tr&\left(\rho \left(\psi_n\boxtimes\ds\frac{\dd\nu}{\dd\mu}\right)(x) \right)\\
&= \tr\left(\rho\dfrac{\dd\nu}{\dd\mu}(x)\right)
\left[\tr\left(\rho\dfrac{\dd\nu}{\dd\mu}(x)\right)\right]^{-1}
\tr\left(\rho\left(\frac{\dd\nu}{\dd\mu}(x)\right)^{1/2}\psi_n(x)\left(\frac{\dd\nu}{\dd\mu}(x)\right)^{1/2}\right)\\
&= \tr\left(\rho\dfrac{\dd\nu}{\dd\mu}(x)\right)
\left[\tr\left(\rho\dfrac{\dd\nu}{\dd\mu}(x)\right)\right]^{-1}
\tr\left(\left(\frac{\dd\nu}{\dd\mu}(x)\right)^{1/2}\rho\left(\frac{\dd\nu}{\dd\mu}(x)\right)^{1/2}\psi_n(x)\right)\\
&= \tr\left(\rho\dfrac{\dd\nu}{\dd\mu}(x)\right)
\tr\left(\left[\tr\left(\rho\dfrac{\dd\nu}{\dd\mu}(x)\right)\right]^{-1}\left(\frac{\dd\nu}{\dd\mu}(x)\right)^{1/2}\rho\left(\frac{\dd\nu}{\dd\mu}(x)\right)^{1/2}\psi_n(x)\right)\\
&= \tr\left(\rho\dfrac{\dd\nu}{\dd\mu}(x)\right)
\tr\left(\tilde\rho_x\psi_n(x)\right).
\end{align*}
Hence, continuity of the trace functional, along with the assumption that $\psi_n$  converges ultraweakly $\mu$-almost surely to $\psi$, yields
\color{black}
\begin{align*}
\lim_{n\to\infty}\tr\left(\rho\left(\psi_n\boxtimes\dfrac{\dd\nu}{\dd\mu}\right)(x)\right)
&=\lim_{n\to\infty}\tr\left(\rho\dfrac{\dd\nu}{\dd\mu}(x)\right)\tr(\tilde{\rho}_x\psi_n(x))\\
&=\tr\left(\rho\dfrac{\dd\nu}{\dd\mu}(x)\right)\tr\left(\tilde{\rho}_x\lim_{n\to\infty}\psi_n(x)\right)\\
&=\tr\left(\rho\dfrac{\dd\nu}{\dd\mu}(x)\right)\tr\left(\tilde{\rho}_x\psi(x)\right)\\
&=\tr\left(\rho\left(\psi\boxtimes\dfrac{\dd\nu}{\dd\mu}\right)(x)\right)
\end{align*}
as required.
\end{proof}

We now prove the main result of this section, namely continuity of quantum expectation, which is a natural quantum analogue of the classical Lebesgue dominated convergence theorem.  
	
\begin{theorem}[Continuity of Quantum Expectation]\label{DCT}
Let $\psi:X \to \BH$. If  $\{\psi_n\}_{n=1}^\infty$ is a sequence of $\nu$-integrable quantum random variables that converges ultraweakly $\mu$-almost surely to $\psi$, and if there exists a $\mu$-integrable random variable $Z:X\to\C$ such that 
$$\ds \left|\tr\left(\rho\left(\psi_n\boxtimes\dfrac{\dd\nu}{\dd\mu}\right)\right)
\right|\leq Z$$
almost surely for all $\rho\in\SH$, then $\psi$ is $\nu$-integrable and $\QE{\nu}{\psi_n}\to\QE{\nu}{\psi}$ ultraweakly.
\end{theorem}

\begin{proof}
Begin by defining the sequence of complex valued random variables $\{\psi_\rho^{(n)}\}_{n=1}^\infty$ by 
$$\psi_\rho^{(n)}=\tr\left(\rho\left(\psi_n\boxtimes\dfrac{\dd\nu}{\dd\mu}\right)\right).$$
Using 
\color{black}
continuity 
\color{black}
of the trace functional along with Lemma~\ref{QConvCor}, we obtain
\begin{align*}
\lim_{n\to\infty}\psi_\rho^{(n)}
=\lim_{n\to\infty}\tr\left(\rho\left(\psi_n\boxtimes\dfrac{\dd\nu}{\dd\mu}\right)\right)
&=\tr\left(\rho\lim_{n\to\infty}\left(\psi_n\boxtimes\dfrac{\dd\nu}{\dd\mu}\right)\right)\\
&=\tr\left(\rho\left(\psi\boxtimes\dfrac{\dd\nu}{\dd\mu}\right)\right).
\end{align*}
That is, $\{\psi_\rho^{(n)}\}_{n=1}^\infty$ converges pointwise $\mu$-almost everywhere to 
$$\ds \tr\left(\rho\left(\psi\boxtimes\dfrac{\dd\nu}{\dd\mu}\right)\right).$$ 
By assumption, the sequence$\{\psi_\rho^{(n)}\}_{n=1}^\infty$ is bounded 
\color{black}
almost surely
\color{black}
 by a $\mu$-integrable random variable $Z:X\to\C$ so by the classical Lebesgue  dominated convergence theorem,
$$\ds \tr\left(\rho\left(\psi\boxtimes\dfrac{\dd\nu}{\dd\mu}\right)\right)$$
is a $\mu$-integrable random variable, and for every $\rho\in\SH$, we have
$$\int_X\tr\left(\rho\left(\psi_n\boxtimes\dfrac{\dd\nu}{\dd\mu}\right)\right)\d\mu \to \int_X\tr\left(\rho\left(\psi\boxtimes\dfrac{\dd\nu}{\dd\mu}\right)\right)\d\mu.$$
Therefore $\psi$ is a $\nu$-integrable function and $\tr(\rho\QE{\nu}{\psi_n})\to\tr(\rho\QE{\nu}{\psi})$ which implies that $\QE{\nu}{\psi_n}\to\QE{\nu}{\psi}$ ultraweakly.  
\end{proof}
		
As a first application of the continuity of quantum expectation we prove that, under certain conditions, quantum expectation is linear over infinite sums.  In fact, this could even be considered as a special case of a quantum Fubini-type theorem.  

\begin{theorem}\label{DCTcor}
Suppose that $\{\psi_n\}_{n=1}^\infty$  is a sequence of $\nu$-integrable quantum random variables.  If
$$\ds \sum_{n=1}^\infty\psi_n=\lim_{N\to\infty}\sum_{n=1}^N\psi_n$$
exists where convergence is with respect to the ultraweak topology of $\BH$, then 
$$\ds \sum_{n=1}^\infty\psi_n$$
is a $\nu$-integrable quantum random variable with
$$\ds \QE{\nu}{\sum_{n=1}^\infty\psi_n}=\sum_{n=1}^\infty \QE{\nu}{\psi_n}.$$
\end{theorem}
	
\begin{proof}
Let $\varphi_N=\ds\sum_{n=1}^N\psi_n$ so that $\varphi_N$ converge ultraweakly $\mu$-almost surely to $\varphi$ where $\varphi=\ds\sum_{n=1}^\infty\psi_n$.  
By Lemma~\ref{QConv}, $\varphi$ is a quantum random variable, and by Theorem~\ref{DCT}, $\varphi$ is $\nu$-integrable and
\begin{equation}\label{eqn1}
\lim_{N\to\infty}\QE{\nu}{\varphi_N}=\QE{\nu}{\varphi}.
 \end{equation}
However,  finite additivity of quantum expectation gives
$$\ds \QE{\nu}{\varphi_N}=\QE{\nu}{\ds\sum_{n=1}^N\psi_n}=\sum_{n=1}^N\QE{\nu}{\psi_n}$$
so from~\eqref{eqn1} we obtain
\[
 \sum_{n=1}^\infty\QE{\nu}{\psi_n} =\lim_{N\to\infty}\sum_{n=1}^N\QE{\nu}{\psi_n} =\lim_{N\to\infty}\QE{\nu}{\varphi_N} =\QE{\nu}{\varphi}
=\QE{\nu}{\sum_{n=1}^\infty\psi_n}
 \]
 as required. 
 \end{proof}

As an example of the type of calculations possible using the previous result, consider the following.
	
\begin{corollary}
If $\psi$ is an effect valued quantum random variable such that $\psi(x)\neq0$ and $\psi(x)\neq1$ for all $x\in X$, then
$\ds  \sum_{n=1}^\infty\QE{\nu}{\psi[1-(1+\psi^{-2})^{-1}]^n\psi}=1$.
\end{corollary}

\color{black}
\begin{proof}
Observe that
\begin{align*}
\sum_{n=1}^\infty [1-(1+\psi^{-2})^{-1}]^n
&=-1+\sum_{n=0}^\infty [1-(1+\psi^{-2})^{-1}]^n\\
&=-1 + ( 1-  [1-(1+\psi^{-2})^{-1}]) ^{-1} \\
&=-1 +  (1+ \psi^{-2}) \\
&= \psi^{-2}
\end{align*}
and so the result now follows from Theorem~\ref{DCTcor}.
\end{proof}
\color{black}


\section{Quantum random variables with quantum expectation zero}\label{MeanZerosect}

We will shortly prove a characterization theorem for quantum random variables with quantum expectation zero. As a preliminary tool, we need the following straightforward lemma.

\begin{lemma}\label{lem1}
If $z\in\BH_+$, then $\ker(z)=\ker(z^{1/2})$ and $\ran(z)=\ran(z^{1/2})$.
\end{lemma}

\begin{proof}
If $\eta\in\ker(z^{1/2})$, then $z^{1/2}\eta=0$ implying that $z\eta=z^{1/2}z^{1/2}\eta=0$ so $\eta\in\ker(z)$.  
Conversely, if $\eta\in\ker(z)$, then $z\eta=0$ so that $0=\langle z\eta,\eta\rangle=\langle z^{1/2}\eta,z^{1/2}\eta\rangle$
implying $z^{1/2}\eta=0$ so $\eta\in\ker(z^{1/2})$. Since $z\in\BH_+$ and $z=z^*$, it  follows from the orthogonal decomposition $\H=\ker(z^*)\oplus\ran(z)$ that $\ran(z)=\ran(z^{1/2})$.  
\end{proof}

\begin{theorem}\label{meanzerothm}
If $\psi:X\to\BH_+$ is a positive $\nu$-integrable quantum random variable, then the following statements are equivalent.
{\em 
\begin{itemize}
\item[(A)]  $\QE{\nu}{\psi}=0$.
\item[(B)]  {\em $\ran(\psi(x))\perp\ran\left(\ds\frac{\dd\nu}{\dd\mu}(x)\right)$ for $\mu$-almost all $x\in X$.}
\item[(C)] {\em $\psi(x)^*\ds\frac{\dd\nu}{\dd\mu}(x)=0$ for $\mu$-almost all $x\in X$.}
\item[(D)] {\em $\left(\psi\boxtimes\ds\frac{\dd\nu}{\dd\mu}\right)(x)=0$ for  $\mu$-almost all $x\in X$.}
\item[(E)]  {\em $\psi(x)^{1/2}\left(\ds\frac{\dd\nu}{\dd\mu}(x)\right)^{1/2}=0$ for $\mu$-almost all $x\in X$.}
\end{itemize}
}
\end{theorem}

\begin{proof} Throughout the proof, let  $z=z(x)$ be given by 
$$\ds z(x)=\psi(x)^{1/2}\left(\frac{\dd\nu}{\dd\mu}(x)\right)^{1/2},$$
 and note that 
$\psi(x)=\psi(x)^*$ since $\psi(x) \in \BH_+$ for all $x\in X$.
To show (E)$\iff$(D), note that $z=0$
\color{black}
for $\mu$-almost all $x$
\color{black}
 if and only if $z^*z=0$ 
\color{black}
for $\mu$-almost all $x$
\color{black}
 and so
\begin{equation}\label{proofeqn1}
\left(\psi\boxtimes\ds\frac{\dd\nu}{\dd\mu}\right)(x)=\left(\frac{\dd\nu}{\dd\mu}(x)\right)^{1/2}\psi(x)\left(\frac{\dd\nu}{\dd\mu}(x)\right)^{1/2}=z^*z \ge 0
\end{equation}
\color{black}
for $\mu$-almost all $x$.
\color{black}
To show (A)$\implies$(E)$\implies$(C), suppose that $\QE{\nu}{\psi}=0$ which implies
\begin{align}\label{eqn2}
\int_X\tr\left(\rho z^*z\right)\d\mu=
\int_X\tr\left(\rho\left(\frac{\dd\nu}{\dd\mu}(x)\right)^{1/2}\psi(x)\left(\frac{\dd\nu}{\dd\mu}(x)\right)^{1/2}\right)\d\mu&=\tr(\rho\QE{\nu}{\psi}) \notag \\
&=0
\end{align}
for every $\rho\in\SH$ and
\color{black}
$\mu$-almost all $x$.
\color{black}
Since $z^*z\ge0$
\color{black}
for $\mu$-almost all $x$
\color{black}, we deduce from~\eqref{eqn2} that 
$\tr\left(\rho z^* z\right)=0$
for every  $\rho\in\SH$
\color{black}
and $\mu$-almost all $x$.
\color{black}
Choosing $\rho=1/d \in \SH$ implies that $\tr(z^*z)=0$
\color{black}
for $\mu$-almost all $x$,
\color{black}
 from which it follows that $z=0$
 \color{black}
for $\mu$-almost all $x$;
\color{black}
that is,  (E) holds. Multiplying (E) on the left by 
$\psi(x)^{1/2}$ and on the right by $\left(\ds\frac{\dd\nu}{\dd\mu}(x)\right)^{1/2}$ yields (C).
To show (B)$\iff$(C)$\implies(D)$, note that if $\psi$ is any $\BH$ valued  (and not just $\BH_+$ valued) quantum random variable, then
$\ds \psi(x)^*\frac{\dd\nu}{\dd\mu}(x)=0$
\color{black}
for $\mu$-almost all $x$
\color{black}
if and only if 
$\ds \left\langle\xi,\psi(x)^*\frac{\dd\nu}{\dd\mu}(x)\eta\right\rangle =0$ for all $\xi$, $\eta\in\H$
\color{black}
and $\mu$-almost all $x$
\color{black}
if and only if 
$ \ds \left\langle\psi(x)\xi,\frac{\dd\nu}{\dd\mu}(x)\eta\right\rangle =0$ for all $\xi$, $\eta\in\H$
\color{black}
and $\mu$-almost all $x$
\color{black}
 if and only if 
$\ds \ran(\psi(x))\perp\ran\left(\frac{\dd\nu}{\dd\mu}(x)\right)$
\color{black}
for $\mu$-almost all $x$.
\color{black}
That is, (B)$\iff$(C). Hence, if (C) holds, then Lemma~\ref{lem1} implies
$\ds \ran(\psi(x))\perp\ran\left(\left(\frac{\dd\nu}{\dd\mu}(x)\right)^{1/2}\right)$
\color{black}
for $\mu$-almost all $x$
\color{black}
and so from the already proved (B)$\iff$(C), we conclude
$\ds \psi(x)^*\left(\frac{\dd\nu}{\dd\mu}(x)\right)^{1/2}=0$
\color{black}
for $\mu$-almost all $x$.
\color{black}
Taking the adjoint of the previous equality and multiplying on the right by 
$\ds \left(\frac{\dd\nu}{\dd\mu}(x)\right)^{1/2}$
yields (D)
as desired.
To complete the proof, we will show (D)$\implies$(A). If $\psi$ is any $\BH$ valued  (and not just $\BH_+$ valued) quantum random variable for which (D) holds,
then since $\QE{\nu}{\psi}$ is the unique operator with
$$\tr(\rho\QE{\nu}{\psi})=\int_X \tr\left(\rho \left(\psi\boxtimes\frac{\dd\nu}{\dd\mu}\right)\right) \d\mu=0$$
for all $\rho\in\SH$, we conclude $\QE{\nu}{\psi}=0$ as required.  
\end{proof}

In the event that $\psi$ is a $\BH$ valued quantum random variable, as opposed to a $\BH_+$ valued one, the statements of the previous theorem are no longer all equivalent.

\begin{corollary}\label{meanzerocor}
Let $\psi:X\to\BH$ be a $\nu$-integrable quantum random variable and consider the following statements.
{\em 
\begin{itemize}
\item[(A)] $\QE{\nu}{\psi}=0$.
\item[(B)] {\em $\ran(\psi(x))\perp\ran\left(\ds\frac{\dd\nu}{\dd\mu}(x)\right)$ for $\mu$-almost all $x\in X$.}
\item[(C)] {\em $\psi(x)^*\ds\frac{\dd\nu}{\dd\mu}(x)=0$ for $\mu$-almost all $x\in X$.}
\item[(D)] {\em $\left(\psi\boxtimes\ds\frac{\dd\nu}{\dd\mu}\right)(x)=0$ for  $\mu$-almost all $x\in X$.}
\end{itemize}
}
The following diagram describes the relationships between these statements.
$$
\begin{array}{ccccccc}
\mathrm{(B)}&\iff&\mathrm{(C)}&\implies&\mathrm{(D)}&\implies&\mathrm{(A)} \\
\end{array}
$$
Moreover, no other implications hold in general.
\end{corollary}

\begin{proof}
The fact that the implications
(B)$\iff$(C)$\implies$(D)
and (D)$\implies$(A)
hold for $\B(\H)$ valued quantum random variables was established in the proof of Theorem~\ref{meanzerothm}. 
To show that no other implications hold in general, we consider two examples. Let $X=\{x_1,x_2\}$, and consider the quantum probability measures $\nu_1$ and $\nu_2$ defined by
\[
\nu_1(\{x_1\})=
\nu_1(\{x_2\})=
\begin{bmatrix}
1/2&0\\
0&1/2\\
\end{bmatrix}
\quad\textrm{and}\quad
\nu_2(\{x_1\})=
\begin{bmatrix}
1&0\\
0&0\\
\end{bmatrix}, \;\;
\nu_2(\{x_2\})=
\begin{bmatrix}
0&0\\
0&1\\
\end{bmatrix}
\]
as well as the quantum random variables $\psi_1$ and $\psi_2$ defined by
\[
\psi_1(x_1)=
\begin{bmatrix}
1&0\\
0&1\\
\end{bmatrix}, \;\;
\psi_1(x_2)=
\begin{bmatrix}
-1&0\\
0&-1\\
\end{bmatrix}
\quad\textrm{and}\quad
\psi_2(x_1)=
\begin{bmatrix}
0&1\\
1&1\\
\end{bmatrix}, \;\;
\psi_2(x_2)=
\begin{bmatrix}
1&1\\
1&0\\
\end{bmatrix}.
\]
Since $X$ is finite, the principal Radon-Nikod\'ym derivative is easily computed, namely
\[
\frac{\dd\nu_i}{\dd\mu_i}(x_j)=2\frac{\nu(\{x_j\})}{\tr(\nu(\{x_j\}))}
\]
for $i,j\in\{1,2\}$.
It is now easy to check that $\QE{\nu_1}{\psi_1}=\begin{bmatrix}
0&0\\
0&0\\
\end{bmatrix}
$ although
$$
\psi_1(x_1)^* \frac{\dd\nu_1}{\dd\mu_1}(x_1)=
\begin{bmatrix}
1&0\\
0&1\\
\end{bmatrix}
\quad\textrm{and}\quad
\psi_1(x_2)^* \frac{\dd\nu_1}{\dd\mu_1}(x_2)=
\begin{bmatrix}
-1&0\\
0&-1\\
\end{bmatrix}
$$
and
$$
\left(\psi_1\boxtimes\frac{\dd\nu_1}{\dd\mu_1}\right)(x_1)=
\begin{bmatrix}
1&0\\
0&1\\
\end{bmatrix}
\quad\textrm{and}\quad\left(\psi_1\boxtimes\frac{\dd\nu_1}{\dd\mu_1}\right)(x_2)=
\begin{bmatrix}
-1&0\\
0&-1\\
\end{bmatrix}.
$$
Hence, in this example (A) holds, but neither (C)  nor (D) hold.
Moreover, one can check that
$$
\left(\psi_2\boxtimes\frac{\dd\nu_2}{\dd\mu_2}\right)(x_1)=
\left(\psi_2\boxtimes\frac{\dd\nu_2}{\dd\mu_2}\right)(x_2)=
\begin{bmatrix}
0&0\\
0&0\\
\end{bmatrix}
$$
whereas
$$
\psi_2(x_1)^* \frac{\dd\nu_2}{\dd\mu_2}(x_1)=
\begin{bmatrix}
0&0\\
2&0\\
\end{bmatrix}
\quad\textrm{and}\quad
\psi_2(x_2)^* \frac{\dd\nu_2}{\dd\mu_2}(x_2)=
\begin{bmatrix}
0&2\\
0&0\\
\end{bmatrix}
$$
providing an example for which (D) holds, but (C) does not hold.
\end{proof}

\begin{corollary}
If $\psi:X\to\BH$ is a $\nu$-integrable quantum random variable and
$\psi(x)\ds\frac{\dd\nu}{\dd\mu}(x)=0$ for $\mu$-almost all $x\in X$, then
$\QE{\nu}{\psi}=0$.
\end{corollary}

\begin{proof}
It follows from the implication (C)$\implies$(A) of Corollary~\ref{meanzerocor} that $\QE{\nu}{\psi^*}=0$ and so $\QE{\nu}{\psi} =\QE{\nu}{\psi^{**}} =\QE{\nu}{\psi^*}^*=0^*=0$ as required.
\end{proof}


\section{A quantum martingale convergence theorem}\label{MCTsect}

In this section we establish a quantum martingale convergence theorem for  quantum martingales obtained by conditioning on a fixed quantum random variable. Recall that a stochastic process $\{M_j\}_{j=0}^\infty$ defined on a filtered probability space is a martingale with respect to the filtration $\{\F_j\}_{j=0}^\infty$ if (i) $M_j$ is $\F_j$-measurable, (ii) $\E[\,|M_j|\,]<\infty$, and (iii) $M_j=\E[M_{j+1}\,|\,\F_j]$  for all $j$. The following version of the martingale convergence theorem is suitable for our purposes;  see Theorem~3.7.3 of Ref.~\cite{bob} for a proof.

\begin{theorem}[Martingale Convergence Theorem]\label{classicMCT}
If $\{M_j\}_{j=0}^\infty$ is a martingale with respect to the filtration $\{\F_j\}_{j=0}^\infty$ and there exists $C>0$ such that $\E[\,|M_j|\,]<C$ for all $j$, then there exists a random variable $M_\infty$ such that $\E[\,|M_\infty|\,]<\infty$ and $M_j$ converges to $M_\infty$ almost surely.  
\end{theorem}

When the martingale is obtained by conditioning on a fixed random variable, the martingale convergence theorem takes the following form; see Corollary~3.6.9 of Ref.~\cite{bob}.

\begin{corollary}
If $\psi$ is a random variable on the filtered probability space $(X,\F, \{\F_j\}_{j=0}^\infty, \mu)$ and satisfies $\E[\,|\psi|\,] <\infty$, then the  martingale $M_j=\E[\psi\,|\,\F_j]$ converges both almost surely and in $L^1(X,\mu)$ to $M_\infty=\E[\psi\,|\,\F_\infty]$ where $\F_\infty=\sigma\left(\bigcup_{j=0}^\infty\F_j\right)$.  If either {\em (i)} $\psi$ is $\F_\infty$-measurable, or  {\em (ii)} $\F_\infty=\F$, then $M_\infty=\psi$.  
\end{corollary}

We now turn our attention to quantum conditional expectation. The following result summarizes the relevant facts from Ref.~\cite{farenick--kozdron2012} that we need; see, in particular, the proof of Theorem~III.1. 

\begin{theorem}\label{condexp}
Suppose that $(X, \borel{X},\nu)$ is a quantum probability space, and that $\psi:X \to \B(\H)_+$ is a $\nu$-integrable quantum random variable with $\QE{\nu}{\psi} \neq 0$.
If $\F(X)$ is a sub-$\sigma$-algebra of $\borel{X}$, then there exists a function $\varphi:X \to \B(\H)$ such that
{\em 
\begin{itemize}
\item[(i)] {\em $\varphi$ is $\F(X)$-measurable,}
\item[(ii)] {\em $\varphi$ is $\nu$-integrable, and }
\item[(iii)]  {\em $\QE{\nu}{\psi \ch{E}} = \QE{\nu}{\varphi \ch{E}}$
for every $E \in \F(X)$.}
\end{itemize}
}
\end{theorem}

We call $\varphi$ a version of quantum conditional expectation of $\psi$ given $\F(X)$ relative to $\nu$ and write
$\varphi = \QCE{\nu}{\psi}{\F(X)}$. Moreover, if $\tilde\varphi$ is any other $\nu$-integrable $\F(X)$-measurable function satisfying 
$\QE{\nu}{\psi \ch{E}} = \QE{\nu}{\tilde \varphi \ch{E}}$ for every $E \in \F(X)$, then 
$\nu(\{x \in X : \varphi(x) \neq \tilde\varphi(x)\}) = 0$. Thus, instead of saying  ``$\varphi = \QCE{\nu}{\psi}{\F(X)}$ $\nu$-almost surely'' we identify different versions and say that $\QCE{\nu}{\psi}{\F(X)}$ is \emph{the}  quantum conditional expectation of $\psi$ given $\F(X)$ relative to $\nu$.
In fact, if $\nu'=\nu|_{\F(X)}$ is the restriction of $\nu$ to $\F(X)$, and 
$$\tilde \nu(E)  = \int_E \psi \d \nu',$$
for $E \in \F(X)$, then $\ds \varphi =  \QCE{\nu}{\psi}{\F(X)}=\frac{\dd \tilde \nu}{\dd \nu'}$.
Clearly, $\varphi:X\to\BH_+$ for $\nu'$-almost all $x\in X$.  Since $\nu'$-measure zero sets have $\nu$-measure zero, 
setting
\[
\varphi(x)= \QCE{\nu}{\psi}{\F(X)}(x)=
\begin{cases}
\ds\frac{\dd\tilde{\nu}}{\dd\nu'}(x), &\text{for } \ds\frac{\dd\tilde{\nu}}{\dd\nu'}(x)\in\BH_+,\\
0,&\textrm{otherwise},\\
\end{cases}
\]
implies $\varphi:X\to\BH_+$ for \emph{all} $x\in X$.
We are now able to prove the important tower property for quantum conditional expectation. Note that this was not considered in Ref.~\cite{farenick--kozdron2012}.  

\begin{theorem}\label{tower}
If  $\psi:X\to\B(\H)_+$ is a $\nu$-integrable quantum random variable with $\QE{\nu}{\psi}\neq0$, and $\F(X)$, $\G(X)$ are sub $\sigma$-algebras of $\borel{X}$ such that $\F(X)\subseteq\G(X)$, then
\begin{equation}\label{towerthmeq}
 \QCE{\nu}{\QCE{\nu}{\psi}{\F(X)}}{\G(X)}=\QCE{\nu}{\psi}{\F(X)}=\QCE{\nu}{\QCE{\nu}{\psi}{\G(X)}}{\F(X)}.
 \end{equation}
\end{theorem}

\begin{proof}
Define $\varphi_f=\QCE{\nu}{\psi}{\F(X)}$ and $\varphi_g=\QCE{\nu}{\psi}{\G(X)}$. To prove the theorem, we will verify that
$\QCE{\nu}{\varphi_f}{\G(X)}=\varphi_f=\QCE{\nu}{\varphi_g}{\F(X)}$.
The first equality in~\eqref{towerthmeq} follows immediately from the fact that $\varphi_f$ is  $\G(X)$-measurable and  $\F(X)\subseteq\G(X)$.
As for the second equality in~\eqref{towerthmeq}, observe that  if $F\in\F(X)$ and $G\in\G(X)$, then $\QE{\nu}{\varphi_f\chi_{F}}=\QE{\nu}{\psi\chi_{F}}$ and $\QE{\nu}{\varphi_{g}\chi_{G}}=\QE{\nu}{\psi\chi_{G}}$, implying
$\QE{\nu}{\varphi_g\chi_{F}}=\QE{\nu}{\psi\chi_{F}}$. 
This, in turn, implies that 
$\QE{\nu}{\varphi_f\chi_{F}}=\QE{\nu}{\varphi_{g}\chi_{F}}$ for any $F\in\F(X)$ yielding $\varphi_{g}=\varphi_f$ as required.
\end{proof}

In analogy with the classical definition, we now state the definition of a quantum martingale.

\begin{defn}
 Let $(X,\borel{X},\nu)$ be a quantum probability space.  A sequence of quantum random variables $\{\varphi_j\}_{j=0}^\infty$ is called a quantum martingale with respect to the filtration  $\{\F_j(X)\}_{j=0}^\infty$  if
\begin{itemize}
\item[(i)] $\varphi_j$ is $\F_j(X)$-measurable for all $j$,
\item[(ii)] $\varphi_j$ is $\nu$-integrable for all $j$, and
\item[(iii)] $\QCE{\nu}{\varphi_{j+1}}{\mathcal{F}_j(X)}=\varphi_j$ for all $j$.
\end{itemize}
\end{defn}
	
It is also important to know that a quantum martingale is obtained by conditioning on a fixed quantum random variable.	
	
\begin{theorem}
If $\psi:X\to\BH_+$ is a $\nu$-integrable quantum random variable and $\QE{\nu}{\psi}\neq0$, then the sequence of $\F_j(X)$-measurable $\nu$-integrable quantum random variables $\{\varphi_j\}_{j=0}^\infty$  where $\varphi_j=\QCE{\nu}{\psi}{\F_j(X)}$ is a quantum martingale.
\end{theorem}
	
\begin{proof}
The fact that $\varphi_j$ is $\F_j(X)$-measurable follows immediately from the definition of conditional expectation. 
The fact that $\varphi_j$ is $\nu$-integrable follows since $\psi$ is $\nu$-integrable and $\QE{\nu}{\varphi_j}=\QE{\nu}{\QCE{\nu}{\psi}{\F_j(X)}}=\QE{\nu}{\psi}$; see Proposition~4.3 of Ref.~\cite{farenick--kozdron2012} for a proof of this fact. 
\color{black}
We now observe that
$$
\QCE{\nu}{\varphi_{j+1}}{\F_j(X)}=\QCE{\nu}{\QCE{\nu}{\psi}{\F_{j+1}(X)}}{\F_j(X)}\\
=\QCE{\nu}{\psi}{\F_j(X)}\\
=\varphi_j$$
where the second equality follows from the tower property, Theorem~\ref{tower}. In other words, 
$$\QCE{\nu}{\varphi_{j+1}}{\F_j(X)} = \varphi_j$$
so that we have proved that $\{\varphi_j\}_{j=0}^\infty$  is, in fact,  a quantum martingale.
\color{black}
\end{proof}
	  
\begin{theorem}[Continuity of Quantum Conditional Expectation]
Consider a quantum probability space  $(X, \borel{X},\nu)$ and suppose that $\F(X)\subseteq\borel{X}$ is a sub $\sigma$-algebra.  Suppose further that
$\{\psi_n\}_{n=1}^\infty$ is a sequence of $\nu$-integrable quantum random variables with $\psi_n:X\to\B(\H)_+$ and $\QE{\nu}{\psi_n}\neq 0$ for all $n$.  If $\psi_n$ converges ultraweakly $\mu$-almost surely to $\psi$,  
then $\QCE{\nu}{\psi_n}{\F(X)}$ converges ultraweakly $\mu$-almost surely to $\QCE{\nu}{\psi}{\F(X)}$.
\end{theorem}
	
\begin{proof}
For any $F\in\F(X)$, we know  $\psi_n\ch{F}$ converges ultraweakly $\mu$-almost surely to $\psi\ch{F}$.
Theorem~\ref{DCT} says that $\QE{\nu}{\psi_n\ch{F}}\to\QE{\nu}{\psi\ch{F}}$ ultraweakly implying that  $\QCE{\nu}{\psi_n}{\F(X)}$ converges ultraweakly $\mu$-almost surely
to $\QCE{\nu}{\psi}{\F(X)}$ as required.  
\end{proof}

\color{black}
We now mention a relationship between the quantum conditional expectation $\QCE{\nu}{\psi}{\F(X)}$  and the family of classical conditional expectations $\QCE{\mu}{\psi_\rho}{\F(X)}$ for $\rho \in \SH$ which follows from the definitions of the objects involved.  Observe that if $\psi:X\to\B(\H)_+$ is a $\nu$-integrable quantum random variable with $\QE{\nu}{\psi}\neq0$, then the following statements are equivalent.
{\em 
\begin{itemize}
\item[{\em (A)}] $\nu(\, \{x \,| \, \varphi(x)=\QCE{\nu}{\psi}{\F(X)}(x)\} \, )=1$.
\item[{\em (B)}] $\mu(\, \{x \, |\, \varphi_\rho(x)=\QCE{\mu}{\psi_\rho}{\F(X)}(x) \;\forall\rho\in\SH\} \, )=1$.
\end{itemize}
}
That is, let $\varphi=\QCE{\nu}{\psi}{\F(X)}$ so that $\varphi$ is a $\F(X)$-measurable quantum random variable with the property that  $\QE{\nu}{\varphi\ch{E}}=\QE{\nu}{\psi\ch{E}}$
for every $E\in\F(X)$. However, this holds if and only if for all $\rho \in \SH$ we have
$\tr(\rho\QE{\nu}{\varphi\ch{E}})=\tr(\rho\QE{\nu}{\psi\ch{E}})$
which in turn holds if and only if $\QE{\mu}{(\varphi\ch{E})_\rho}=\QE{\mu}{(\psi\ch{E})_\rho}$. However, 
$(\varphi\ch{E})_\rho=\varphi_\rho\ch{E}$  so that $\QE{\mu}{\varphi_\rho\ch{E}}=\QE{\mu}{\psi_\rho\ch{E}}$. Therefore,
$\varphi_\rho=\QCE{\mu}{\psi_\rho}{\F(X)}$
\color{black}

We are now in a position to prove the main result of this paper, namely a quantum martingale convergence theorem for the quantum martingale $\varphi_j=\QCE{\nu}{\psi}{\F_j(X)}$.
Although we will prove that the sequence $\{\varphi_j\}_{j=0}^\infty$ has a unique limit, in contrast to the classical situation, the value of the limiting random variable $\varphi_\infty$ cannot be determined in general. In fact, all that can be said is that $\varphi_\infty$ and $\QCE{\nu}{\psi}{\F_\infty(X)}$ differ by a quantum random variable $\Phi$ satisfying $\Phi_\rho=0$ for all $\rho\in\SH$.  

\begin{theorem}[Quantum Martingale Convergence Theorem]\label{MCT}
Consider a quantum probability space $(X,\borel{X},\nu)$ with filtration $\{\F_j(X)\}_{j=0}^\infty$ and let $\psi:X\to\BH_+$ be a $\nu$-integrable quantum random variable with $\QE{\nu}{\psi}\neq 0$.  Consider the quantum martingale $\varphi_j=\QCE{\nu}{\psi}{\F_j(X)}$.  There exists a $\nu$-integrable quantum random variable $\varphi_\infty$ such that 
{\em 
\begin{itemize}
\item[(i)] {\em $\varphi_j$ converges  ultraweakly $\mu$-almost surely to $\varphi_\infty$,}
\item[(ii)] {\em $\varphi_\infty$ is $\F_\infty(X)=\sigma\left(\bigcup_{j=0}^\infty \F_j(X)\right)$-measurable, and }
\item[(iii)] {\em $\varphi_\infty\in\{\, \QCE{\nu}{\psi}{\F_\infty}+\Phi\ |\ \Phi_\rho=0\ \forall\rho\in\SH \, \}$.}
\end{itemize}
}
Furthermore, if either 
{\em 
\begin{itemize}
\item[(iv)] {\em $\F_\infty(X)=\borel{X}$, or }
\item[(v)] {\em $\psi$ is $\F_\infty(X)$-measurable, }
\end{itemize}
}
then  $\varphi_\infty\in\{\, \psi+\Phi\ |\ \Phi_\rho=0\ \forall\rho\in\SH\, \}$.  
\end{theorem}

\begin{proof}
For every $\rho \in \SH$, since $\varphi_j$ is $\nu$-integrable it follows that $\varphi_{j_\rho}$ is $\mu$-integrable and satisfies
\[
\QE{\mu}{\,\left|\varphi_{j_\rho}\right|\,}=\QE{\mu}{\,\left|\,\QCE{\mu}{\psi_\rho}{\F_j(X)}\,\right|\,}\leq\QE{\mu}{\,\left|\psi_\rho\right|\,}
\]
for all $j$. By the martingale convergence theorem, Theorem~\ref{classicMCT}, for every $\rho \in \state{\H}$ there exists a $\mu$-integrable $\tilde{\varphi}_{\infty_\rho}$ such that 
\begin{itemize}
 \item[(i)] $\varphi_{j_\rho}$ converges to  $\tilde{\varphi}_{\infty_\rho}$ almost surely,
 \item[(ii)] $\tilde{\varphi}_{\infty_\rho}$ is $\F_\infty(X)=\sigma\left(\bigcup_{j=0}^\infty\F_j(X)\right)$-measurable, and 
 \item[(iii)] $\tilde{\varphi}_{\infty_\rho}=\QCE{\mu}{\psi_\rho}{\F_\infty(X)}$.
\end{itemize}
But this implies that  $\varphi_j$ converges  ultraweakly $\mu$-almost surely to some $\varphi_\infty$ with $\varphi_{\infty_\rho}=\tilde{\varphi}_{\infty_\rho}$ for all $\rho\in\SH$.  By the continuity of quantum expectation, it follows that $\varphi_\infty$ is $\nu$-integrable.  
Let $\tilde{\varphi}=\QCE{\nu}{\psi}{\F_\infty(X)}$ so that 
\[
\tilde{\varphi}_{\rho}=\QCE{\mu}{\psi_\rho}{\F_\infty(X)}=\tilde{\varphi}_{\infty_\rho} =\varphi_{\infty_\rho}.			
\]
However, if $\Phi$ is another $\nu$-integrable quantum random variable with $\Phi_\rho=0$ for all $\rho\in\SH$, then
$\left(\tilde{\varphi}+\Phi\right)_\rho=\tilde{\varphi}_{\rho}+\Phi_\rho=\tilde{\varphi}_{\rho}=\varphi_{\infty_\rho}$
implying
\[
\varphi_\infty\in\{\, \QCE{\nu}{\psi}{\F_\infty}+\Phi\ |\ \Phi_\rho=0\ \forall\rho\in\SH \, \}
\]
as claimed. Finally, if either $\F_\infty(X)=\borel{X}$ or $\psi$ is $\F_\infty(X)$-measurable, then $\QCE{\nu}{\psi}{\F_\infty(X)}=\psi$  so that  $\varphi_\infty\in\{\, \psi+\Phi\ |\ \Phi_\rho=0\ \forall\rho\in\SH\, \}$ as required.  
 \end{proof}

We will now study the set of possible limits from our quantum martingale convergence theorem.  

\begin{theorem}\label{MCTlim}
 Let $(X,\borel{X},\nu)$ be a quantum probability space and let $\psi:X\to\BH_+$ be a $\nu$-integrable quantum random variable. Define the set 
\[
\Gamma_{\nu,\psi}=\{\, \Psi \,|\, \Psi=\QCE{\nu}{\psi}{\F_\infty(X)}+\Phi\textrm{ with }\Phi_\rho=0\ \forall\rho\in\SH \, \}.
\]
If $\Psi_1\in\Gamma_{\nu,\psi}$ then $\Psi_2\in\Gamma_{\nu,\psi}$ if and only if
$\ds (\Psi_2-\Psi_1)\boxtimes\frac{\dd\nu}{\dd\mu}=0$.
\end{theorem}

\begin{proof}
Let $\Psi_1$, $\Psi_2\in\Gamma_{\nu,\psi}$ so that $\Psi_{1_\rho}=\Psi_{2_\rho}$ for all $\rho\in\SH$.  Therefore,
\[
0=\tr\left(\rho\left(\Psi_2\boxtimes\frac{\dd\nu}{\dd\mu}\right)\right)-\tr\left(\rho\left(\Psi_1\boxtimes\frac{\dd\nu}{\dd\mu}\right)\right) =\tr\left(\rho\left((\Psi_2-\Psi_1)\boxtimes\frac{\dd\nu}{\dd\mu}\right)\right).
\]
Since this equality holds for all $\rho\in\SH$, it follows that $\ds (\Psi_2-\Psi_1)\boxtimes\frac{\dd\nu}{\dd\mu}=0$ 
as required. Following the same reasoning in reverse gives the theorem.
\end{proof}

We can now use our results from Section~\ref{MeanZerosect} to study $\Gamma_{\nu,\psi}$.  We know that if $\Phi$ is a quantum random variable then $\Phi_\rho=0$ implies $\QE{\nu}{\Phi}=0$ whereas the converse is not necessarily true.  

\begin{corollary}
If 
$$\Sigma_{\nu,\psi}=\{\, \Psi \, | \, \Psi=\QCE{\nu}{\psi}{\F_\infty(X)}+\Phi,\ \QE{\nu}{\Phi}=0\, \},$$ 
then $\Gamma_{\nu,\phi}\subseteq\Sigma_{\nu,\psi}$.
\end{corollary}

\begin{proof}
Suppose that $\Psi\in\Gamma_{\nu,\psi}$.  Then $\Psi=\QCE{\nu}{\psi}{\F_\infty(X)}+\Phi$ where $\Phi_\rho=0$ for all $\rho\in\SH$.  Then by the earlier remark, it follows that $\QE{\nu}{\Phi}=0$, so that $\Psi=\QCE{\nu}{\psi}{\F_\infty(X)}+\Phi$ with $\QE{\nu}{\Phi}=0$.  Hence $\Psi\in\Sigma_{\nu,\psi}$ and $\Gamma_{\nu,\psi}\subseteq\Sigma_{\nu,\psi}$ as required.  
\end{proof}

%
\section*{Acknowledgements}
%

\color{black}
Much of this research was done by the first author in his master's thesis~\cite{kylerthesis} under the supervision of the second author. The work of the second author is supported, in part, by the Natural Sciences and Engineering Research Council of Canada. The second author also thanks the Isaac Newton Institute for Mathematical Sciences, Cambridge, for its hospitality during the Random Geometry programme where the initial writing of this paper was done. Finally, special thanks are due to both Doug Farenick and Sarah Plosker for many valuable discussions about this, and related, material.
\color{black}


\end{document}